\documentclass[submission,copyright,creativecommons]{eptcs}
\providecommand{\event}{Non-Classical Logics 2024 (NCL'24)}

\usepackage{iftex}

\ifpdf
  \usepackage{underscore}         
  \usepackage[T1]{fontenc}        
\else
  \usepackage{breakurl}           
\fi

\usepackage{amsmath,amssymb,latexsym}
\input amssym.def
\input amssym.tex

\newtheorem{theorem}{Theorem}
\newtheorem{lemma}[theorem]{Lemma}

\newtheorem{corollary}[theorem]{Corollary}

\newtheorem{defn}[theorem]{Definition}
\newtheorem{eg}[theorem]{Example}
\newtheorem{rem}[theorem]{Remark}

\newenvironment{definition}{\defn \em }{\vspace{0.3 em}} 
\newenvironment{example}{\eg \em }{\vspace{0.3 em}} 
\newenvironment{remark}{\rem \em }{\vspace{0.3 em}} 

\newenvironment{proof}
{\vspace{0.5 em} \noindent {\bf Proof}}{\QED {\vspace{0.5 em}}}

\def\filledbox{\vrule height 1.8ex width .8ex depth 0ex }

\newcommand{\qed}{\filledbox}

\def\QED{\ifmmode\squareforqed\else{\unskip\nobreak\hfil
\penalty50\hskip1em\null\nobreak\hfil\qed
\parfillskip=0pt\finalhyphendemerits=0\endgraf}\fi}

\newcommand{\ignore}[1] { }

\newcommand{\bsm}[1]{\mbox{\boldmath \scriptsize$#1$}}

\newcommand{\F}{\mathcal F}
\renewcommand{\P}{\mathcal P}

\newcommand{\bfour}{\mbox{\boldmath$4$}}

\newcommand{\bB}{\mbox{\boldmath$B$}}

\newcommand{\bC}{\mbox{\boldmath$C$}}

\newcommand{\bD}{\mbox{\boldmath$D$}}

\newcommand{\bE}{\mbox{\boldmath$E$}}

\newcommand{\bL}{\mbox{\boldmath$L$}}

\newcommand{\bT}{\mbox{\boldmath$T$}}

\newcommand{\bV}{\mbox{\boldmath$V$}}

\newcommand{\bv}{\mbox{\boldmath$v$}}

\newcommand{\sbl}{\bsm L}

\newcommand{\bsigma}{\mbox{\boldmath$\Sigma$}}

\renewcommand{\t}{t}

\newcommand{\sbsigma}{\bsm \Sigma}

\newcommand{\mvk}{\textbf{\emph{mv-K}}}
\newcommand{\mvd}{\textbf{\emph{mv-D}}}
\newcommand{\mvt}{\textbf{\emph{mv-T}}}
\newcommand{\mvkfour}{\textbf{\emph{mv-K4}}}
\newcommand{\mvsfour}{\textbf{\emph{mv-S4}}}
\newcommand{\mvb}{\textbf{\emph{mv-B}}}
\newcommand{\mvsfive}{\textbf{\emph{mv-S5}}}
\newcommand{\mvil}{\textbf{\emph{mv\hspace{0.1 em}I\hspace{0.1 em}L}}}

\title{Many-Valued Modal Logic}
\author{Amir Karniel
\institute{Department of Mathematics\\
Technion - Israel Institute of Technology\\
Haifa 3200003,
Israel}
\and
Michael Kaminski
\institute{Department of Computer Science\\
Technion - Israel Institute of Technology\\
Haifa 3200003,
Israel}
}
\def\titlerunning{Many-Valued Modal Logic}
\def\authorrunning{A. Karniel \& M. Kaminski}
\begin{document}
\maketitle

\begin{abstract}
We combine the concepts of modal logics and
many-valued logics in a general and comprehensive way.
Namely,
given any finite linearly ordered set of truth values and
any set of propositional connectives defined by truth tables,
we define the many-valued minimal normal modal logic,
presented as a Gentzen-like sequent calculus, and
prove its soundness and strong completeness with respect to
many-valued Kripke models.
The logic treats necessitation and possibility independently,
i.e., they are not defined by each other,
so that the duality between them is reflected in the proof system itself.
We also prove the finite model property (that implies strong decidability)
of this logic and consider some of its extensions.
Moreover, we show that there is exactly one way to define negation
such that De Morgan's duality between necessitation and possibility holds.
In addition,
we embed many-valued intuitionistic logic into
one of the extensions of our many-valued modal logic.
\end{abstract}

\section{Introduction}
\label{s: introduction}

The (two-valued) logic~K is the minimal normal modal logic.
It extends
classical propositional calculus with the modal connective~$\Box$,
the rule of inference
\begin{equation}
\label{eq: nec}
\frac{\textstyle \varphi}{\textstyle \Box \varphi}
\end{equation}
and the axiom scheme
\begin{equation}
\label{eq: M1}
\Box(\varphi \supset \psi) \supset (\Box\varphi \supset \Box\psi)
\end{equation}
Semantically,~K is characterized by
Kripke models~\cite{Kripke59,Kripke63}.

In this paper we define
a many-valued counterpart~\mvk{} of~K,
in which
the necessity connective is interpreted as the infimum of
all relevant values and
the possibility connective is interpreted as their supremum, and
nothing is assumed about the underlying propositional connectives.
Syntactically, our proof system is an extension of that
in~\cite{KaminskiF21} to the modal case.
The possibility connective~$\Diamond$ is treated explicitly.
The reason for such a treatment is that, in \mvk{},
$\Box$ and $\Diamond$ are not necessarily interdefinable.
This is because our set of connectives
does not necessarily contain negation, and even if it does,
nothing is assumed about its truth table.
We also show extensions of \mvk{},
which are counterparts of some well-known extensions of~K.
We establish the finite model property of \mvk{} and its extensions.
We then show the unique definition of negation such that
De Morgan's duality between $\Box$ and $\Diamond$ holds.
Finally, we prove that many-valued intuitionistic logic is a fragment
of one of the extensions of \mvk{}.

A number of many-valued normal modal logics is known from the literature.
In \cite{Ostermann88},
an $n$-valued modal logic is based on
the {\L}ukasiewicz classical $n$-valued connectives.
The paper contains Hilbert-style calculi for the generalizations of
the two-valued normal modal logics~T,S4, and~S5, and
the author notes that other generalizations are also possible.
This work seems to generalize~\cite{SchotchJLM78},
in which three-valued modal {\L}ukasiewicz logics are considered.
Three-valued modal logics with different connectives are considered
in \cite{Segerberg67}.

In~\cite{Thomason78} and~\cite{Morgan79},
general notions of many-valued modal logics are suggested,
using designated values and a rather general interpretation of
the modal connective.

In~\cite{Morikawa89},
proof systems relying on matrices are presented for
normal three-valued modal logics
based on any arbitrary set of propositional connectives.
Among other logics,
the three-valued counterparts of
the two-valued~T,S4, and~S5 are presented.

In~\cite{Fitting91}, the author presents a sequent calculus for
modal logics based on any finite lattice of truth values.
These logics, in addition to all propositional constants,
have all classical propositional connectives.
The semantics relies on a many-valued accessibility relation,
that is further discussed in~\cite{Fitting92}.
This paper also addresses the possibility connective~$\Diamond$
that is treated explicitly, because, in general,
$\Diamond$ and $\Box$ are not interdefinable.\footnote{
As noted above, these connectives are not interdefinable in
our paper either, but for a different reason.}
In addition, some extensions of the many-valued modal logics are mentioned
at the end of~\cite{Fitting92}.

The most general approach (for our purposes) was, probably, taken
in~\cite{Takano94},
where
a proof system, relying on matrices of labelled formulas,
is presented for many-valued normal modal logics with
an arbitrary set of propositional connectives.
The system is appropriate for any semantic interpretation of
the necessity connective satisfying certain conditions~--
not only for its interpretation as the infimum.
The system is weakly complete and possesses the subformula property
(that implies weak decidability).
However,
the possibility connective is not addressed in~\cite{Takano94} at all, and
extensions of the logic are not presented there.

Another general approach is taken in \cite{Ferm98},
where proof systems using tableaux are suggested for a variety of
finite-valued modal logics with generalized modalities.

In \cite{Bou11}, counterparts of K, using Hilbert-style proof systems,
are presented for any finite residuated lattice of truth values,
allowing many-valued accessibility relations.
The logics address only the necessity operator
(it is only mentioned that the possibility operator should, in general,
be addressed separately and not as an abbreviation of $\neg\Box\neg$), and
they are based on a fixed set of propositional connectives.
The semantics use designated values to interpret validity of formulas.

In \cite[Chapter 9.1]{Metcalfe08},
many-valued modal logics are discussed, referring also to Gentzen systems,
logic extensions and logic embeddings.
Again, a fixed set of connectives is assumed and
the semantic interpretation is algebraic.

Our research introduces a novel and comprehensive framework for
many-valued modal logics that stands out by
integrating several key features simultaneously:
the use of an arbitrary ``base logic'',
the use of Gentzen-like sequents of labelled formulas,
the independent treatment of both necessity and possibility modalities,
the demonstration of strong completeness and strong decidability, and
addressing all basic logic extensions.
While each of these elements has been explored individually in
previous studies, our work combines them into a single coherent system.
This combination allows for a more robust and flexible logical framework
that can handle a wider variety of logical scenarios and applications.
By employing labeled formulas
(discussed, e.g., in \cite{Baaz98} and \cite{KaminskiF21}),
we can address any truth value rather
than being limited to designated ones,
providing a significant advantage in terms of expressive power.
The motivation behind this research lies in the importance of
many-valued modal logics in contexts
with inherent uncertainty or gradations of truth,
such as fuzzy logic systems and multi-agent systems.
Additionally, extending these logics to include features like
transitive accessibility relations is crucial for modeling
more complex systems.
Our results include the finite model property,
ensuring the logics' strong decidability, and
embedding many-valued intuitionistic logic within our framework,
thus offering a comprehensive and robust tool for logical analysis.

The paper is organized as follows.
In Section~\ref{s: logic},
we introduce a many-valued modal logic~\mvk{} and
present a sound and strongly complete\footnote{
That is, complete with respect to the consequence relation.}
proof system for it.
Section~\ref{s: proof} deals with the {\em canonical model} theorem and
the proof of the strong completeness of~\mvk{}. 
Section~\ref{s: extensions} contains
some extensions of~\mvk{} and
their soundness and completeness
with respect to certain classes of Kripke models and,
in Section~\ref{s: fmp},
we explain why \mvk{} and
its extensions from Section~\ref{s: extensions}
possess the finite model property\footnote{
Thus, they are strongly decidable.}.
Then, in Section \ref{s: duality},
we present the appropriate definition of negation
so that~$\Box$ and~$\Diamond$ are interdefinable.
Finally,
in Section \ref{s: intuitionistic},
we embed many-valued intuitionistic logic in
our many-valued counterpart of~S4.

We conclude this section with the note that, because of
the limitation on the publications length,
a number of proofs is omitted.

\section{Many-valued modal logic}
\label{s: logic}

In this section we define a many-valued logic, \mvk{},
assuming a linear order on the set of truth values.

In what follows,
$\bV = \{ \bv_1,\ldots,\bv_n \}$, $n \geq 2$,
is a set of truth values ordered by
\[
\bv_1 < \bv_2 < \cdots < \bv_n
\]

Formulas of \mvk{} are built
from propositional variables by means of propositional connectives
(of arbitrary arities) and the modal connectives $\Box$ and $\Diamond$.
The set of all \mvk{} formulas will be denoted by~$\F$.
The semantics of propositional connectives is given by
truth tables, where, as usual, the truth table of
an $\ell$-ary propositional connective $\ast$ is a function
$\ast : \bV ^\ell \rightarrow \bV$
and
the semantics of the modal connectives is given below.

A {\em labelled} formula is a pair $(\varphi , k)$,
where $\varphi$ is a formula and $k = 1,\ldots,n$.
The intended meaning of such a labelled formula is that
$\bv_k$ is the truth value associated with $\varphi$.

Sequents are expressions of the form $\Gamma \rightarrow \Delta$,
where $\Gamma$ and $\Delta$ are finite (possibly empty) sets
of labelled formulas and
$\rightarrow$ is not a symbol of the underlying language.

The \mvk{} semantics is as follows.

A \emph{many-valued Kripke model}
(or \emph{many-valued K-model} or just Kripke model)
is a triple $M = \langle W,R,I \rangle$,
where
\begin{itemize}
\item
$W$ is a nonempty set (of possible worlds),
\item
$R$ is a
binary (accessibility) relation on $W$, and
\item
${I : W \times \P \rightarrow \bV}$,
where $\P$ is the set of propositional variables,
is a (valuation) function.
\end{itemize}

For a world $u \in W$,
we define the set of {\em successors} of $u$,
denoted by $S(u)$, as
\[
S(u) = \{v \in W : u R v \}
\]
and
extend $I$ to $W\times \F$, recursively, as follows.

\begin{itemize}
\item
$I(u,\ast(\varphi_1,\ldots,\varphi_\ell))
=
\ast(I(u,\varphi_1),\ldots,I(u,\varphi_\ell))$,
\item
$I(u,\Box\varphi)
=
\inf(\{ I(v,\varphi) : v \in S(u) \})$,
where $\inf(\emptyset)$ is $\bv_n$, and
\item
$I(u,\Diamond\varphi)
=
\sup(\{I(v,\varphi) : v\in S(u)\})$,
where $\sup(\emptyset)$ is $\bv_1$.
\end{itemize}

Note that,
if $S(u) \neq \emptyset$,
then, since $\bV$ is finite and linearly ordered,
$\inf(\{I(v,\varphi) : v \in S(u)\})$ and
$\sup(\{I(v,\varphi) : v\in S(u)\})$
are, actually,
$\min(\{I(v,\varphi) : v \in S(u)\})$ and
$\max(\{I(v,\varphi) : v \in S(u)\})$,
respectively.

We also write $M,u \models (\varphi,k)$, if $I(u,\varphi) = \bv_k$.

The satisfiability relation $\models$ between worlds of $W$ and
sequents of labelled formulas is defined as follows.

A world $u$ satisfies a sequent $\Gamma \rightarrow \Delta$,
denoted $M,u \models \Gamma\rightarrow\Delta$, if the following holds.
\begin{itemize}
\item
If for each $(\varphi , k) \in \Gamma$, $I(u,\varphi) = \bv_k$,
then for some $(\varphi , k) \in \Delta$,
$I(u,\varphi) = \bv_k$.\footnote{
\label{f: meta}
In other words, $v$ satisfies a sequent $\Gamma \rightarrow \Delta$,
if the metavalue of the classical metasequent
$\{ I(u,\varphi) = \bv_k : (\varphi,k) \in \Gamma \}
\rightarrow
\{ I(u,\varphi) = \bv_k : (\varphi,k) \in \Delta \}$
is ``true.''}
\end{itemize}
A Kripke model $M$ satisfies a sequent $\Gamma \rightarrow \Delta$,
if each world in $W$ satisfies $\Gamma \rightarrow \Delta$ and
$M$ satisfies a set of sequents $\bsigma$,
if it satisfies each sequent in $\bsigma$.
Finally,
a set of sequents $\bsigma$ {\em semantically entails}
a sequent $\Gamma \rightarrow \Delta$,
denoted $\bsigma \models \Gamma \rightarrow \Delta$,
if each many-valued Kripke model
satisfying $\bsigma$ also satisfies $\Gamma \rightarrow \Delta$.

Let $i$ and $j$ be nonnegative integers.
We denote the set of integers between $i$ and $j$ by $[i,j]$.
That is
\[
[i,j] = \{ k : i \leq k \leq j\}
\]
In particular, if $i>j$, $[i,j]$ is empty.

By definition,
\[
\overline{[i,j]}
=
\{ 1,\ldots,n \} \setminus [i,j]
=
[1,i - 1] \cup [j+1,n]
\]

For convenience, we define
\[
(\varphi,k)^+ = \{\varphi\} \times [k,n]
\]
and
\[
(\varphi,k)^- = \{\varphi\} \times [1,k]
\]

\begin{definition}
\label{d: x}
For a set of labelled formulas $\Gamma$,
the set of labelled formulas $\Gamma^\times$ is defined as follows.
\[
\Gamma^\times
=
\bigcup
\{ \{ \psi \} \times \overline{[i_\psi,j_\psi]} :
(\Box\psi,i_\psi),(\Diamond\psi,j_\psi) \in \Gamma \}
\]
That is,
for all $\psi$ such that
$(\Box\psi,i_\psi),(\Diamond\psi,j_\psi)\in \Gamma$,
$\Gamma^\times$ includes the set
$\{\psi\} \times \overline{[i_\psi,j_\psi]}$ and nothing more.
\end{definition}

The idea lying behind the definition of $\Gamma^\times$ is,
that in a Kripke model $M$,
if $uRv$ and $u$ satisfies \emph{every} element of $\Gamma$,
then $v$ satisfies \emph{no} element of $\Gamma^\times$.

We define next the proof system of \mvk{}.

The axioms are:
\begin{equation}
\label{eq: ordinary axiom}
(\varphi , k) \rightarrow (\varphi , k)
\end{equation}
and
\begin{equation}
\label{eq: table axiom}
(\varphi_1 , k_1),\ldots,(\varphi_\ell , k_\ell)
\rightarrow
(\ast(\varphi_1,\ldots,\varphi_\ell),k)
\end{equation}
for each table entry $\bv_{k_1},\ldots,\bv_{k_\ell}$ such that
$\ast(\bv_{k_1},\ldots,\bv_{k_\ell}) = \bv_k$,
and
the logical rules are

\par \noindent
\begin{equation}
\label{eq: mvk 1}
\frac
{\textstyle
(\varphi,k) \rightarrow \Gamma^\times
}
{\textstyle
(\Box\varphi,k),\Gamma \rightarrow}
\, \ \ \ \ k\neq n
\end{equation}
\begin{equation}
\label{eq: mvk 2}
\frac
{\textstyle
(\varphi,k) \rightarrow \Gamma^\times
}
{\textstyle
(\Diamond\varphi,k),\Gamma \rightarrow}
\, \ \ \ \ k\neq 1
\end{equation}
\par \noindent
and
the structural rules below.

\par \noindent
$k$-left-shift:
\begin{equation}
\label{eq: mvml ls}
\frac
{\textstyle
\Gamma , (\varphi , k) \rightarrow \Delta}
{\textstyle
\Gamma
\rightarrow
\Delta , \{ \varphi \} \times \overline{\{ k \}}}
\end{equation}
$k^\prime,k^{\prime\prime}$-right-shift:
\begin{equation}
\label{eq: mvml rs}
\frac
{\textstyle
\Gamma \rightarrow \Delta , (\varphi , k^\prime)}
{\textstyle
\Gamma , (\varphi , k^{\prime\prime}) \rightarrow \Delta}
\, \ \ \ \ k^\prime\neq k^{\prime\prime}
\end{equation}
$k$-left-weakening:
\begin{equation}
\label{eq: mvml lw}
\frac
{\textstyle
\Gamma \rightarrow \Delta}
{\textstyle
\Gamma , (\varphi ,k) \rightarrow \Delta}
\end{equation}
$k$-right-weakening:
\begin{equation}
\label{eq: mvml rw}
\frac
{\textstyle
\Gamma \rightarrow \Delta}
{\textstyle
\Gamma \rightarrow \Delta , (\varphi ,k)}
\end{equation}
$k$-cut:
\begin{equation}
\label{eq: mvml cut}
\frac
{\textstyle
\Gamma^\prime \rightarrow
\Delta^\prime , (\varphi , k)
\ \ \ \ \ \
\Gamma^{\prime\prime} , (\varphi , k)
\rightarrow
\Delta^{\prime\prime}}
{\textstyle
\Gamma^\prime , \Gamma^{\prime\prime}
\rightarrow
\Delta^\prime , \Delta^{\prime\prime}}
\end{equation}
$k^\prime,k^{\prime\prime}$-resolution:
\begin{equation}
\label{eq: mvml resolution}
\frac
{\textstyle
\Gamma^\prime \rightarrow \Delta^\prime , (\varphi , k^\prime)
\ \ \ \ \ \
\Gamma^{\prime\prime}
\rightarrow
\Delta^{\prime\prime} , (\varphi , k^{\prime\prime})
}
{\textstyle
\Gamma^\prime , \Gamma^{\prime\prime}
\rightarrow
\Delta^\prime , \Delta^{\prime\prime}}
\, \ \ \ \ k^\prime\neq k^{\prime\prime}
\end{equation}

In fact, cut and resolution are derivable from each other,
see~\cite[Proposition~3.3]{KaminskiF21}.

We shall also need the two following derivable rules.
One is
``multi-shift''
\begin{equation}
\label{eq: mvml multi-shift}
\frac
{\textstyle
\{ \Gamma_k , (\varphi , k) \rightarrow \Delta_k : k \in K \}}
{\textstyle
\bigcup_{k \in K} \Gamma_k
\rightarrow \bigcup_{k \in K} \Delta_k ,
\{ \varphi \} \times \overline{K}}
\end{equation}
for $K \subset \{ 1,\ldots,n \}$ and
$\overline{K} = \{ 1,\ldots,n \} \setminus K$,
see~\cite[Remark~3.5]{KaminskiF21}, and
the other is its generalization
\begin{equation}
\label{eq: mvml super-multi-shift}
\hspace{-2.2 em}
\frac
{\textstyle
\{ \Gamma_{k_1,\ldots,k_\ell} ,
(\varphi_1 , k_1),\ldots,(\varphi_\ell,k_\ell)
\rightarrow
\Delta_{k_1,\ldots,k_\ell} :
k_1 \in K_1,\ldots,k_\ell \in K_\ell \}}
{\textstyle
\bigcup_{k_1\in K_1,\ldots,k_\ell \in K_\ell} \Gamma_{k_1,\ldots,k_\ell}
\rightarrow
\bigcup_{k_1\in K_1,\ldots,k_\ell \in K_\ell} \Delta_{k_1,\ldots,k_\ell} ,
\{ \varphi_1 \} \times \overline{K_1},\ldots,
\{ \varphi_\ell \} \times \overline{K_\ell}}
\end{equation}

The derivation of~\eqref{eq: mvml super-multi-shift} is rather long and
is omitted.\footnote{
A skeptical reader can easily verify that this rule is valid and
then add it to \mvk{}.}

We precede the statement of the soundness and completeness theorem
for \mvk{} with a number of examples.

\begin{example}
Sequents
\label{e: 31}
\begin{equation}
\label{eq: box leq diamond}
(\Box\varphi , k) \rightarrow (\Diamond\varphi, k)^+\, \ \ \ \ k \neq n
\end{equation}
and
\begin{equation}
\label{eq: diamond geq box}
(\Diamond\varphi , k) \rightarrow (\Box\varphi, k)^-\, \ \ \ \ k \neq 1
\end{equation}
are \mvk{} derivable.

The derivation of~\eqref{eq: box leq diamond} is as follows,
where, in steps $2_j$ and $3_j$, $j < k$.

\[
\hspace{-1.1 em}
\begin{array}{lll}
1. &
(\varphi,k) \rightarrow (\varphi,k)
&
\mbox{axiom~\eqref{eq: ordinary axiom}}
\\
2. & 
(\varphi,k) \rightarrow \{\varphi\} \times [1,n]
&
\mbox{follows from 1 by
$n - 1$ right weakenings~\eqref{eq: mvml rw}}
\\
2_j. &
(\varphi,k) \rightarrow \{\varphi\} \times \overline{[k,j]}
&
\mbox{follows from 2, because, for $j < k$, $[k,j] = \emptyset$}
\\
3_j. &
(\Box\varphi,k),(\Diamond\varphi,j) \rightarrow
&
\mbox{follows from $2_j$ by~\eqref{eq: mvk 1} with
$\Gamma$ being $\{(\Box\varphi,k),(\Diamond\varphi,j)\}$}
\\
4. &
(\Box\varphi , k) \rightarrow (\Diamond\varphi, k)^+
&
\mbox{follows from all $3_j$, $j < k$,
by multi-shift~\eqref{eq: mvml multi-shift}}
\end{array}
\]

The derivation of~\eqref{eq: diamond geq box} is dual to that
of~\eqref{eq: box leq diamond} and is omitted.
\end{example}

\begin{example}
\label{e: 37}
Sequent
\begin{equation}
\label{eq: box is n}
(\Box\varphi , n) \rightarrow (\Diamond\varphi, 1),(\Diamond\varphi, n)
\end{equation}
is \mvk{} derivable.

Indeed, for $k \neq 1,n$,
by~$k$ $n,j$-right-shifts~\eqref{eq: mvml rs}, $j \leq k$,
on~\eqref{eq: diamond geq box}, we obtain
\[
(\Diamond\varphi,k),(\Box\varphi,n)\rightarrow \, \ \ \ \
k = 2,3,\ldots, n - 1
\]
from which~\eqref{eq: box is n}
follows by multi-shift~\eqref{eq: mvml multi-shift}.
\end{example}

\begin{example}
Sequents
\label{e: 33}
\begin{equation}
\label{eq: dead end}
(\Box\varphi,n),(\Diamond\varphi,1)\rightarrow(\Box\psi,n) 
\end{equation}
and
\begin{equation}
\label{eq: dead end 2}
(\Box\varphi,n),(\Diamond\varphi,1)\rightarrow(\Diamond\psi,1)
\end{equation}
are \mvk{} derivable.

The derivation of~\eqref{eq: dead end} is as follows,
where, in steps $2_i$ and $3_i$, $i \neq n$.
\[
\hspace{-1.1 em}
\begin{array}{lll}
1. & 
\rightarrow \{\varphi\} \times [1,n]
&
\mbox{derivable sequent, see~\cite[Proposition~3.4]{KaminskiF21}}
\\
2_i. & 
(\psi,i) \rightarrow \{\varphi\} \times [1,n]
&
\mbox{follows from 1 by $i$-left-weakening~\eqref{eq: mvml lw}}
\\
3_i. &
(\Box\psi,i),(\Box\varphi,n),(\Diamond\varphi,1)\rightarrow
&
\mbox{follows from $2_i$ by~\eqref{eq: mvk 1}, with $\Gamma$ being}
\\
& &
\mbox{$\{(\Box\varphi,n),(\Diamond\varphi,1)\}$,
because $\overline{[n,1]} = [1,n]$}
\\
4. &
(\Box\varphi,n),(\Diamond\varphi,1)\rightarrow(\Box\psi,n)
&
\mbox{follows from all $3_i$, $i \neq n$,
by multi-shift~\eqref{eq: mvml multi-shift}}
\end{array}
\]

The derivation of~\eqref{eq: dead end 2} is dual to that
of~\eqref{eq: dead end} and is omitted.
\end{example}

\begin{example}
In this example we show that the sequent
\[
(\Box (p\supset q),3),(\Box p,3)\rightarrow (\Box q,3)
\]
is derivable in the modal extension of
the $\L$ukasiewicz three-valued logic.
That is,
$n=3$ and the truth table of implication $\supset$ is as follows.
\[
\begin{array}{c|c|c|c}
\supset & 1 & 2 & 3 
\\
\hline
1 & 3 & 3 & 3
\\
2 & 2 & 3 & 3
\\
3 & 1 & 2 & 3
\end{array}
\]

In steps $3_k,4_k,5_k,7_k$ and $9_k$ of the proof below, $k \in \{1,2\}$.
\[
\hspace{-0.5 em}
\begin{array}{lll}
1. & 
(p,3),(q,1) \rightarrow (p\supset q, 1)
&
\mbox{axiom~\eqref{eq: table axiom}}
\\
2. & 
(p,3),(q,2) \rightarrow (p\supset q, 2)
&
\mbox{axiom~\eqref{eq: table axiom}}
\\
3_k. & 
(p,3),(q,k),(p\supset q, 3) \rightarrow 
&
\mbox{follows from
either 1 or 2 by right-shift~\eqref{eq: mvml rs}}
\\
4_k. & 
(q,k) \rightarrow \{p\}\times [1,2],\{p\supset q\}\times [1,2]
&
\mbox{follows from
$3_k$ by left-shifts \eqref{eq: mvml ls}}
\\
5_k. & 
(\Box q,k),(\Box p,3),(\Diamond p,3),(\Box (p\supset q),3),
(\Diamond (p\supset q),3)\rightarrow
&
\mbox{follows from
$4_k$ by \eqref{eq: mvk 1}, because $[1,2]=\overline{[3,3]}$}
\\
6. & 
(\Box p,3)\rightarrow (\Diamond p,3),(\Diamond p,1)
&
\mbox{\eqref{eq: box is n}}
\\
7_k. & 
(\Box q,k),(\Box p,3),(\Box (p\supset q),3),(\Diamond (p\supset q),3)
\rightarrow (\Diamond p,1)
&
\mbox{follows from $5_k$ and 6 by cut~\eqref{eq: mvml cut}}
\\
8. & 
(\Box (p\supset q),3)
\rightarrow
(\Diamond (p\supset q),3),(\Diamond (p\supset q),1)
&
\mbox{\eqref{eq: box is n}}
\\
9_k. & 
(\Box q,k),(\Box p,3),(\Box p\supset q,3)
\rightarrow
(\Diamond p,1),(\Diamond (p\supset q),1)
&
\mbox{follows from $7_k$ and 8 by cut~\eqref{eq: mvml cut}}
\\
10. & 
(\Box p,3),(\Box p\supset q,3)
\rightarrow
(\Diamond p,1),(\Diamond (p\supset q),1),(\Box q,3)
&
\mbox{follows from all $9_k$ by multi-shift~\eqref{eq: mvml multi-shift}}
\\
11. & 
(\Box p,3),(\Diamond p,1)\rightarrow (\Box q,3)
&
\mbox{\eqref{eq: dead end}}
\\
12. &
(\Box p,3),(\Box p\supset q,3)
\rightarrow
(\Diamond (p\supset q),1),(\Box q,3)
&
\mbox{follows from 10 and 11 by
cut~\eqref{eq: mvml cut}}
\\
13. & 
(\Box (p\supset q),3),(\Diamond (p\supset q),1)\rightarrow (\Box q,3)
&
\mbox{\eqref{eq: dead end}}
\\
14. &
(\Box p,3),(\Box p\supset q,3)\rightarrow (\Box q,3)
&
\mbox{follows from 12 and 13 by
cut~\eqref{eq: mvml cut}}
\end{array}
\]
\end{example}

\begin{theorem}
\label{t: main}
Let $\bsigma$ and $\Gamma \rightarrow \Delta$ be
a set of sequents and a sequent, respectively.
Then $\bsigma \vdash \Gamma \rightarrow \Delta$
if and only if
$\bsigma \models \Gamma \rightarrow \Delta$.
\end{theorem}

The proof of the ``only if'' part of theorem (soundness) is
by induction on the derivation length, and
the proof of the ``if'' part of theorem (strong completeness)
is rather involved and
follows from the canonical model theorem in the next section.

\section{The canonical model theorem and
the proof of the ``if'' part of Theorem~\protect{\ref{t: main}}}
\label{s: proof}

For the proof of the strong completeness of \mvk{},
we extend the definition of provability to
infinite sets of labelled formulas.

For a set of sequents $\bsigma$,
a (not necessarily finite) set of labelled formulas $\Gamma$, and
a finite set of labelled formulas $\Delta$,
we write $\bsigma \vdash \Gamma\rightarrow\Delta$,
if there exists a finite subset $\Gamma^{\prime}$ of $\Gamma$ such that
$\bsigma\vdash\Gamma^{\prime} \rightarrow\Delta$.

A set of labelled formulas $\Gamma$ is called \emph{$\bsigma$-consistent},
if $\bsigma \not\vdash\Gamma \rightarrow$.

A set of sequents $\bsigma$ is called \emph{consistent},
if $\bsigma \not\vdash \rightarrow$.\footnote{
\label{f: equivalent}
Equivalently, $\bsigma$ is consistent,
if there exists a $\bsigma$-consistent set of formulas $\Gamma$.}

\begin{lemma} 
\label{l: zorn}
If $\bsigma \not\vdash \Gamma \rightarrow \Delta$,
then there exists a maximal (with respect to inclusion)
$\bsigma$-consistent set $\Gamma^\prime$ including $\Gamma$
such that $\bsigma \not\vdash \Gamma^\prime \rightarrow \Delta$.
\end{lemma}

The proof is straightforward, by Zorn's lemma, and is omitted.

\begin{lemma}
\label{l: maximal set property}
{\em (\cite[Lemma~3.12 and the following observation]{KaminskiF21})}
If\, $\Gamma$ is a maximal set for which
$\bsigma \not\vdash \Gamma\rightarrow\Delta$,
then for every formula $\varphi$ there exists a unique
$k \in \{1,\ldots,n\}$ such that $(\varphi,k) \in \Gamma$.
\end{lemma} 

From now on, we enumerate the set of all formulas $\F$ as
$\psi_1,\psi_2,\ldots$.

For a consistent set of sequents $\bsigma$,
the \emph{$\bsigma$-canonical model}
$M_{\sbsigma}
=
\langle
W_{\sbsigma},R_{\sbsigma},I_{\sbsigma}
\rangle$
is defined as follows.

\begin{itemize}
\item
$W_{\sbsigma}$ is the set of all maximal $\bsigma$-consistent sets.
Since $\bsigma$ is consistent, by Lemma~\ref{l: zorn}, 
$W_{\sbsigma}$ is nonempty.
\item
For
worlds $\Gamma^\prime,\Gamma^{\prime\prime} \in W_{\sbsigma}$,
$\Gamma^\prime R_{\sbsigma} \Gamma^{\prime\prime}$
if and only if
for the unique $i_1,i_2,\ldots$, $j_1,j_2,\ldots$, and~$k_1,k_2,\ldots$
such that
\[
(\Box\psi_1,i_1),(\Diamond\psi_1,j_1),(\Box\psi_2,i_2),
(\Diamond\psi_2,j_2),\ldots \in \Gamma^\prime
\]
and
\[
(\psi_1,k_1),(\psi_2,k_2),\ldots \in \Gamma^{\prime\prime}
\]
provided by Lemma~\ref{l: maximal set property},
$i_m \leq k_m\leq j_m$ for all $m = 1,2,\ldots$\footnote{
Equivalently,
$\Gamma^\prime R_{\sbsigma}\Gamma^{\prime\prime}$
if and only if
$(\Gamma^\prime)^\times \cap \Gamma^{\prime\prime} = \emptyset$. }
\item
For $u \in W_{\sbsigma}$ and $p \in \P$,
$I_{\sbsigma}(u,p)$ is the unique value $\bv_k$ such that $(p,k)\in u$.
\end{itemize}

\begin{theorem}
\label{t: main theorem}
{\em (The canonical model theorem)}
For all labelled formulas $(\varphi,k)$ and all $u \in W_{\sbsigma}$,
$(\varphi,k) \in u$ if and only if $M_{\sbsigma},u \models (\varphi,k)$.
\end{theorem}


For the proof of Theorem~\ref{t: main theorem} we need the lemma below.

\begin{lemma}
\label{l: i leq j}
Let $\Gamma$ be a $\bsigma$-consistent set of formulas and let
\begin{equation}
\label{eq: in gamma}
(\Box \varphi^\prime ,k) (\Box \psi ,i) , (\Diamond \psi , j) \in \Gamma
\end{equation}
where $k \neq n$.
Then $i \leq j$.
\end{lemma} 

\begin{proof}
Assume to the contrary that $i > j$.
We distinguish among the cases of, $i \neq n$, $j \neq 1$, and
$i = n$ and $j = 1$.

If $i \neq n$, then
\[
\begin{array}{lll}
1. & 
\Gamma \rightarrow (\Box \psi,i)
&
\mbox{follows from
axiom~\eqref{eq: ordinary axiom}, with $\varphi$ being $\Box \psi$ and}
\\
& &
\mbox{$k$ being $i$, and~\eqref{eq: in gamma}}
\\
2. & 
\Gamma \rightarrow (\Diamond \psi,j)
&
\mbox{follows from
axiom~\eqref{eq: ordinary axiom}, with $\varphi$ being $\Diamond \psi$
and}
\\
& &
\mbox{$k$ being $j$, and~\eqref{eq: in gamma}}
\\
3. & 
(\Box \psi,i) \rightarrow (\Diamond \psi,i)^+
&
\mbox{\eqref{eq: box leq diamond}
with $\varphi$ being $\psi$ and $k$ being $i$}
\\
4. & 
\Gamma \rightarrow (\Diamond \psi,i)^+
&
\mbox{follows from 1 and 3 by cut}
\\
5. &
\Gamma \rightarrow
&
\mbox{follows from 2 and 4 by
$n - i$ resolutions~\eqref{eq: mvml resolution}}
\end{array}
\]

However,
$\bsigma\vdash \Gamma \rightarrow$ contradicts
the $\bsigma$-consistency of $\Gamma$.
	
The case of $j \neq 1$ is dual to that of $i \neq n$ and is omitted.

Let $i = n$ and $j = 1$.
Then
\[
\begin{array}{lll}
1. & 
\Gamma \rightarrow (\Box \varphi^\prime,k)
&
\mbox{follows from
axiom~\eqref{eq: ordinary axiom},
with $\varphi$ being $\Box \varphi^\prime$}
\\
2. & 
\Gamma \rightarrow (\Box \psi,n)
&
\mbox{follows from
axiom~\eqref{eq: ordinary axiom}, with $\varphi$ being $\Box \psi$ and}
\\
& &
\mbox{$k$ being $n$, and~\eqref{eq: in gamma}}
\\
3. & 
\Gamma \rightarrow (\Diamond \psi,1)
&
\mbox{follows from
axiom~\eqref{eq: ordinary axiom}, with $\varphi$ being $\Diamond \psi$
and}
\\
& &
\mbox{$k$ being $1$, and~\eqref{eq: in gamma}}
\\
4. & 
(\Box \psi,n),(\Diamond \psi,1) \rightarrow (\Box \varphi^\prime,n)
&
\mbox{\eqref{eq: dead end}
with $\varphi$ being $\psi$ and $\psi$ being $\varphi^\prime$}
\\
5. & 
\Gamma \rightarrow (\Box \varphi^\prime,n)
&
\mbox{follows from 2,\,3, and 4 by two cuts}
\\
6. &
\Gamma \rightarrow
&
\mbox{follows from 1 and 5 by resolution~\eqref{eq: mvml resolution}}
\end{array}
\]

Again,
$\bsigma\vdash \Gamma \rightarrow$ contradicts
the $\bsigma$-consistency of $\Gamma$.
\end{proof}

\begin{proof}
\textbf{of Theorem~\ref{t: main theorem}}
It is sufficient to prove the ``only if'' part of the theorem, i.e.,
that $(\varphi,k) \in u$ implies ${M_{\sbsigma},u\models (\varphi,k)}$.
This is because, if $(\varphi,k) \notin u$,
then, by Lemma~\ref{l: maximal set property},
$(\varphi,k^\prime)\in u$ for $k^\prime\neq k$.
Therefore by the ``only if'' part of the theorem,
$M_{\sbsigma},u \models (\varphi,k^\prime)$,
implying ${M_{\sbsigma},u \nvDash (\varphi,k)}$.

The proof is by induction on the complexity of $\varphi$.
For the cases of
an atomic formula and a propositional principal connective,
see~\cite[Proposition~3.13]{KaminskiF21}.

Let $\varphi$ be of the form $\Box\varphi^\prime$ and
assume that for some $i_1,i_2,\ldots$ and $j_1,j_2,\ldots$,
\[
(\Box\psi_1,i_1),(\Diamond\psi_1,j_1),
(\Box\psi_2,i_2),(\Diamond\psi_2,j_2),\ldots \in u
\]
We distinguish between the cases of $k \neq n$ and $k = n$.
\vspace{0.2 em}

\par \noindent
$\bullet$
Let $k\neq n$.
By the induction hypothesis and
the definition of $R_{\sbsigma}$,
for each world $v \in S(u)$,
$\bv_k \leq I_{\sbsigma}(v,\varphi^\prime)$.\footnote{
This is because, for some $m = 1,2,\ldots,$,
$\varphi^\prime$ is $\psi_m$.}
Therefore, for the proof of
\[
\bv_k
=
\min(\{I_{\sbsigma}(v,\varphi^\prime) : v\in S(u)\})
=
I_{\sbsigma}(u,\varphi)
\]
it suffices to show that
\[
\bv_k \in \{ I_{\sbsigma}(v,\varphi^\prime) : v\in S(u)\}
\]
i.e., that there exist a world $v \in W_{\sbsigma}$ and
$k_1 \in [i_1,j_1],k_2 \in [i_2,j_2],\ldots$ such that
\[
(\varphi^\prime,k),(\psi_1,k_1),(\psi_2,k_2),\ldots\in v
\]
This is because, by definition of $R_{\sbsigma}$, $uR_{\sbsigma}v$ and,
by the induction hypothesis, $I_{\sbsigma}(v,\varphi^\prime)=v_k$.

By Lemma~\ref{l: zorn}, for existence of such $k_1,k_2,\ldots$ and $v$,
it suffices to show that there exist
$k_1\in [i_1,j_1],k_2\in [i_2,j_2],\ldots$ such that
the set of labelled formulas
\begin{equation}
\label{eq: sigma consistent}
\{ (\varphi^\prime,k),(\psi_1,k_1),(\psi_2,k_2),\ldots \}
\end{equation}
is $\bsigma$-consistent.

For the proof, assume to the contrary that for all
$k_1\in [i_1,j_1],k_2\in [i_2,j_2],\ldots$,~\eqref{eq: sigma consistent}
is
$\bsigma$-inconsistent.
That is,
\begin{equation}
\label{eq: set of sequents}
\bsigma
\vdash
(\varphi^\prime,k),(\psi_1,k_1),(\psi_2,k_2),\ldots \rightarrow
\, \ \ \ \
k_1\in [i_1,j_1],k_2\in [i_2,j_2],\ldots
\end{equation}
Note that, by Lemma~\ref{l: i leq j}, $i_m \leq j_m$,
for all $m = 1,2,\ldots$.
Thus, the set of sequents in~\eqref{eq: set of sequents} is nonempty.

We contend that there exists a non-negative integer $L$ such that
\begin{equation}
\label{eq: set of sequents L}
\begin{array}{c}
\bsigma
\vdash
(\varphi^\prime,k),(\psi_1,k_1),(\psi_2,k_2),\ldots,(\psi_L,k_L)
\rightarrow
~~~~~~~~~~~~~~~~~~~~
\\
\hspace{5.5 em}
k_1 \in [i_1,j_1],k_2 \in [i_2,j_2],\ldots,k_L \in [i_L,j_L]
\end{array}
\end{equation}
Then we shall apply
rules~\eqref{eq: mvml super-multi-shift} and~\eqref{eq: mvk 1}
to the set of sequents in~\eqref{eq: set of sequents L}.

For the proof of our contention,
we consider a tree $T$ whose nodes are sets of
labelled formulas of the form
\begin{equation}
\label{eq: nodes}
\begin{array}{c}
\{
(\psi_1,k_1),(\psi_2,k_2),\ldots,(\psi_m,k_m) \}
~~~~~~~~~~~~~~~~~~~~~~~~~~~~~~~~
\\
\hspace{5.5 em}
k_1 \in [i_1,j_1],k_2 \in [i_2,j_2],\ldots,k_m \in [i_m,j_m] , \
\end{array}
\end{equation}
$m = 0,1,\ldots$, such that each node \eqref{eq: nodes} is $\bsigma$-consistent when $(\varphi^\prime,k)$ is added to it as an element,
and
the successors of a node~\eqref{eq: nodes} are nodes of the form
\[
\{
(\psi_1,k_1),(\psi_2,k_2),\ldots,
(\psi_m,k_m),(\psi_{m+1},k_{m+1})
\}
\]
where $k_{m+1} \in [i_{m+1},j_{m+1}]$.

Thus, nodes~\eqref{eq: nodes} are of height $m$.
In particular, the root of $T$ is $\emptyset$,
if $\{(\varphi^\prime,k)\}$ is $\bsigma$-consistent.
Otherwise, $T$ is empty.

This tree $T$ is of a finite branching degree, because a node
of height $m$ has at most $j_{m+1}-i_{m+1}+1$ successors.
Also, $T$ has no infinite paths.
Indeed, an infinite path would correspond to a choice of
$k_1\in [i_1,j_1],k_2\in [i_2,j_2],\ldots$.
However, the set of labelled formulas
$\{(\varphi^\prime,k),(\psi_1,k_1),(\psi_2,k_2),\ldots\}$
is $\bsigma$-inconsistent.
Thus, the path contains a node that becomes $\bsigma$-inconsistent,
when $(\varphi^\prime,k)$ is added to it,
in contradiction with the definition of $T$.
Therefore,
by the contraposition of the K\"{o}nig infinite lemma~\cite{Konig26},
$T$ is finite.

Let $H$ be the height of $T$ ($H$ is defined as $-1$, if $T$ is empty).
Then, for $L = H + 1$, we have~\eqref{eq: set of sequents L},
which proves our contention.

Now, from~\eqref{eq: set of sequents L},
by~\eqref{eq: mvml super-multi-shift} we obtain
\[
\bsigma\vdash (\varphi^\prime,k)
\rightarrow
\{\psi_1\} \times \overline{[i_1,j_1]},
\{\psi_2\} \times \overline{[i_2,j_2]},
\ldots,
\{\psi_L\}\times\overline{[i_L,j_L]}
\]
from which, by~\eqref{eq: mvk 1} we obtain
\[
\bsigma
\vdash
(\Box\varphi^\prime,k),
(\Box\psi_1,i_1),(\Diamond\psi_1,j_1),
(\Box\psi_2,i_2),(\Diamond\psi_2,j_2),
\ldots,
(\Box\psi_L,i_L),(\Diamond\psi_L,j_L)
\rightarrow
\]
that contradicts the $\bsigma$-consistency of $u$.
\vspace{0.2 em}

\par \noindent
$\bullet$
Let $k= n$.
If $S(u) = \emptyset$, then, trivially,
$M_{\sbsigma},u \models (\Box \varphi,k)$.
Otherwise,
by the induction hypothesis and the definition of $R_{\sbsigma}$,
for all worlds $v \in S(u)$,
we have $\bv_n \leq I_{\sbsigma}(v,\varphi^\prime)$,
implying $I_{\sbsigma}(v,\varphi^\prime) = \bv_n$.
Thus,
\[
\min(\{I_{\sbsigma}(v,\varphi^\prime) : v \in S(u)\})
=
\min(\{\bv_n\})
=
\bv_n
\]
and $M_{\sbsigma},u \models (\varphi,k)$ follows.
\vspace{0.3 em}

The case of $\Diamond$ is dual to that of $\Box$.
We just replace $\Box$ with $\Diamond$, $\Diamond$ with $\Box$,
$\min$ with $\max$, $n$ with $1$, and $\leq$ with $\geq$.
We leave the details to the reader.
\end{proof}

\begin{corollary}
\label{c: sigma}
We have $M_{\sbsigma} \models \bsigma$.
\end{corollary}

\begin{proof}
Let $u \in W_{\sbsigma}$ and let $\Gamma \rightarrow \Delta \in \bsigma$.
Assume that $u$ satisfies all labelled formulas in $\Gamma$ and
assume to the contrary that $u$ satisfies
no labelled formula in ${\Delta}$.
By Theorem \ref{t: main theorem}, ${\Gamma}\subseteq u$, and
for all $(\varphi,k_\varphi) \in \Delta$ there is
$k_\varphi^\prime \neq k_\varphi$ such that
$(\varphi,k_\varphi^\prime) \in u$.
By definition, $\bsigma \vdash \Gamma \rightarrow \Delta$.
Therefore,
by $k_\varphi,k_\varphi^\prime$-right-shifts~\eqref{eq: mvml rs},
\[
\bsigma
\vdash
\Gamma ,
\{(\varphi,k_\varphi^\prime) : (\varphi,k_\varphi)\in{\Delta}\}
\rightarrow
\]
implying $\bsigma \vdash u \rightarrow$, because
$\Gamma ,
\{(\varphi,k_\varphi^\prime) : (\varphi,k_\varphi)\in{\Delta}\}
\subseteq
u$.
This, however, contradicts $\bsigma$-consistency of $u$.
\end{proof}

\begin{proof}
\textbf{of the ``if'' part of Theorem~\ref{t: main}}
Assume $\bsigma \not\vdash \Gamma\rightarrow\Delta$.
By Lemma \ref{l: zorn}, there exists a maximal set $\Gamma^\prime$
including $\Gamma$ such that
$\bsigma \not\vdash \Gamma^\prime\rightarrow\Delta$.
By the definition of $M_{\sbsigma}$, $\Gamma^\prime \in W_{\sbsigma}$.
We contend that $M_{\sbsigma} \not\models \Gamma \rightarrow\Delta$.
Namely, $M_{\sbsigma},\Gamma^\prime \not\models \Gamma\rightarrow\Delta$. 

Since $\Gamma \subseteq\Gamma^\prime$,
by Theorem~\ref{t: main theorem},
$\Gamma^\prime$ satisfies all labelled formulas in $\Gamma$.
However, it satisfies no labelled formula in $\Delta$, because,
otherwise, by Theorem~\ref{t: main theorem},
such a formula would belong to $\Gamma^\prime$, implying
$\bsigma \vdash \Gamma^\prime \rightarrow \Delta$,
in contradiction with the definition of $\Gamma^\prime$.
Thus, $M_{\sbsigma} \not\models \Gamma \rightarrow\Delta$,
which completes the proof of our contention and,
together with Corollary~\ref{c: sigma},
completes the proof of the ``if'' part of the theorem.
\end{proof}

\section{Extensions of \mvk{}}
\label{s: extensions}

In this section, $\bL$ is an extension of \mvk{} with additional axioms.

We write $\bsigma \vdash_{\sbl} \Gamma \rightarrow \Delta$,
if $\Gamma \rightarrow \Delta$ is derivable from $\bsigma$ in $\bL$
(and we keep writing $\bsigma \vdash \Gamma \rightarrow \Delta$,
if $\bL$ is \mvk{} itself).
We generalize this notation to sequents $\Gamma \rightarrow \Delta$ with
an infinite antecedent $\Gamma$, like in the previous section.

Clearly, the results of the previous section apply also to
any extension $\bL$.
Below, we just rewrite them with respect to $\bL$.

\begin{definition}
\label{d: l-sigma-consistent}
A set of labeled formulas $\Gamma$ is called
\emph{$\bL$-$\bsigma$-consistent},
if ${\bsigma \not\vdash_{\sbl} \Gamma\rightarrow}$.
\end{definition}

\begin{definition}
\label{d: l-consistent}
A set of sequents $\bsigma$ is called \emph{L-consistent},
if $\bsigma \not\vdash_{\sbl} \rightarrow$, or,
equivalently, if there exists an $\bL$-$\bsigma$-consistent set,
cf. footnote~\ref{f: equivalent}.
\end{definition}

\begin{lemma}
\label{l: zorn l}
{\em (Cf. Lemma~\ref{l: zorn}.)}
If $\bsigma \not\vdash_{\sbl} \Gamma\rightarrow\Delta$,
then there exists a maximal $\bL$-$\bsigma$-consistent set
$\Gamma^\prime$ including $\Gamma$ such that
$\bsigma \not\vdash_{\sbl} \Gamma^\prime \rightarrow \Delta$.
\end{lemma}

\begin{lemma}
\label{l: maximal set property l}
{\em (Cf. Lemma~\ref{l: maximal set property}.)}
If $\Gamma$ is a maximal set for which
$\bsigma\not\vdash_{\sbl}\Gamma\rightarrow\Delta$,
then for every formula $\varphi \in \F$ there exists
a unique $k \in \{1,\ldots,n\}$ such that $(\varphi,k)\in \Gamma$.
\end{lemma}

For an $\bL$-consistent set of sequents $\bsigma$,
we define the \emph{$\bL$-$\bsigma$-canonical model}
$M_{\sbl,\sbsigma}
=
\langle
W_{\sbl,\sbsigma},
 R_{\sbl,\sbsigma},
I_{\sbl,\sbsigma}
\rangle$
just like the $\bsigma$-canonical model $M_{\sbsigma}$
in Section~\ref{s: proof},
except that $W_{\sbl,\sbsigma}$ is the set of
all maximal $\bL$-$\bsigma$-consistent sets.
Note that $W_{\sbl,\sbsigma}$ is nonempty,
because $\bsigma$ is $\bL$-consistent. 

\begin{corollary}
\label{c: extensions}
For an $\bL$-consistent set of sequents $\bsigma$ the following holds.
\begin{itemize}
\item[$(i)$]
{\em (Cf. Theorem~\ref{t: main theorem}.)}
For all labelled formulas $(\varphi,k)$ and all $u \in W_{\sbl,\sbsigma}$,
$(\varphi,k) \in u$
if and only if
$M_{\sbl,\sbsigma},u \models (\varphi,k)$.
\item[$(ii)$]
{\em (Cf. Corollary~\ref{c: sigma}.)}
$M_{\sbl,\sbsigma} \models \bsigma$.
\item[$(iii)$]
{\em (Cf. the ``if'' part of Theorem~\ref{t: main}.)}
If\, $\bsigma \not\vdash_{\sbl} \Gamma\rightarrow\Delta$,
then $M_{\sbl,\sbsigma} \not\models \Gamma\rightarrow\Delta$. 
\end{itemize}
\end{corollary}

We proceed with some extensions of \mvk{} which are sound and
strongly complete for the many-valued Kripke models defined below.

\begin{definition}
\label{d: serial-eucledian}
A binary relation $R \subseteq W \times W$
is called \emph{serial} (or \emph{with no dead-ends}),
if for all $u \in W$, $S(u) \neq \emptyset$, and
is called {\em Eucledian},
if for all $u,v,w \in W$, $uRv$ and $uRw$ imply $vRw$.
\end{definition}

\begin{definition}
A many-valued Kripke model $M=\langle W,R,I \rangle$ is called
\emph{serial/ reflexive/ transitive/ symmetric/ Euclidean},
if the accessibility relation $R$ is
serial/ reflexive/ transitive/ symmetric/ Euclidean,
respectively. 
\end{definition}

In this section,
the many-valued modal logics, we shall deal with,
result from \mvk{} by adding
some subsets of the following axioms.
\begin{equation}
\label{eq: mvml total}
(\Box\varphi,n) \rightarrow (\Diamond\varphi,n)
\end{equation}
\vspace{-0.3 em}
\begin{equation}
\label{eq: mvml reflexive 1}
(\Box\varphi,k) \rightarrow (\varphi,k)^+
\end{equation}
\vspace{-0.3 em}
\begin{equation}
\label{eq: mvml reflexive 2}
(\varphi,k) \rightarrow (\Diamond\varphi,k)^+
\end{equation}
\vspace{-0.3 em}
\begin{equation}
\label{eq: mvk transitive 1}
(\Box\varphi,k)\rightarrow (\Box\Box\varphi,k)^+
\end{equation}
\vspace{-0.3 em}
\begin{equation}
\label{eq: mvk transitive 2}
(\Diamond\Diamond\varphi,k)\rightarrow (\Diamond\varphi,k)^+
\end{equation}
\vspace{-0.3 em}
\begin{equation}
\label{eq: mvk symmetric 1}
(\varphi,k)\rightarrow (\Box\Diamond\varphi,k)^+
\end{equation}
\vspace{-0.3 em}
\begin{equation}
\label{eq: mvk symmetric 2}
(\Diamond\Box\varphi,k)\rightarrow (\varphi,k)^+
\end{equation}
\vspace{-0.3 em}
\begin{equation}
\label{eq: mvk euclidean 1}
(\Diamond\varphi,k)\rightarrow (\Box\Diamond\varphi,k)^+
\end{equation}
\vspace{-0.3 em}
\begin{equation}
\label{eq: mvk euclidean 2}
(\Diamond\Box\varphi,k)\rightarrow (\Box\varphi,k)^+
\end{equation}

\begin{theorem}
\label{t: mvk extensions completeness}
Let $L$ be an extension of {\em \mvk{}} and
let $\bsigma$ be an $L$-consistent set of sequents.
\begin{enumerate}
\item[$(i)$]
If \eqref{eq: mvml total} is an axiom of $L$,
then $M_{L,\sbsigma}$ is serial.
\item[$(ii)$]
If \eqref{eq: mvml reflexive 1} and \eqref{eq: mvml reflexive 2}
are axioms of $L$,
then $M_{L,\sbsigma}$ is reflexive.
\item[$(iii)$]
If \eqref{eq: mvk transitive 1} and \eqref{eq: mvk transitive 2}
are axioms of $L$,
then $M_{L,\sbsigma}$ is transitive.
\item[$(iv)$]
If \eqref{eq: mvk symmetric 1} and \eqref{eq: mvk symmetric 2}
are axioms of $L$,
then $M_{L,\sbsigma}$ is symmetric.
\item[$(v)$]
If \eqref{eq: mvk euclidean 1} and \eqref{eq: mvk euclidean 2}
are axioms of $L$,
then $M_{L,\sbsigma}$ is Euclidean.
\end{enumerate}
\end{theorem}


Next we define the the many-valued counterparts of
the two-valued modal logics D,\,T,\,K4,\,S4,\,B, and S5.

\begin{definition}
\label{d: extensions}
~

\begin{itemize}
\item
The many-valued modal logic \mvd{} is
obtained from \mvk{} by adding to it~\eqref{eq: mvml total}.
\item
The many-valued modal logic \mvt{} is
obtained from \mvk{} by adding
to it~\eqref{eq: mvml reflexive 1} and~\eqref{eq: mvml reflexive 2}.
\item
The many-valued modal logic \mvkfour{} is
obtained from \mvk{} by adding
to it~\eqref{eq: mvk transitive 1} and~\eqref{eq: mvk transitive 2}.
\item
The many-valued modal logic \mvsfour{} is
obtained from \mvt{} by adding
to it~\eqref{eq: mvk transitive 1} and~\eqref{eq: mvk transitive 2}.
\item
The many-valued modal logic \mvb{} is
obtained from \mvk{} by adding
to it~\eqref{eq: mvk symmetric 1} and~\eqref{eq: mvk symmetric 2}.
\item
The many-valued modal logic \mvsfive{} is
obtained from \mvt{} by adding
to it~\eqref{eq: mvk euclidean 1} and~\eqref{eq: mvk euclidean 2}.
\end{itemize}
\end{definition}

The above many-valued logics, but \mvd{} are defined by pairs of axioms~-
the many valued counterpart of the two-valued one and its dual,
because the logics under consideration
do not necessarily have negation.
Thus, unlike in the two-valued case,
$\Box$ and $\Diamond$ are not interdefinable.
We address the extension of these logics with negation in
Section~\ref{s: duality}.

Note that the above axioms are many-valued counterparts of axioms~$\bD$,
see~\cite[p.~29]{HughesC84},~$\bT$,\,$\bfour$,\,$\bB$,
see~\cite[p.~10]{HughesC84}, and~$\bE$, see~\cite[p.~11]{HughesC84}.

\begin{theorem}
\label{t: extensions}
~

\begin{itemize}
\item[$(i)$]
{\em \mvd{}} is sound and (strongly) complete
with respect to serial Kripke models.
\item[$(ii)$]
{\em \mvt{}} is sound and (strongly) complete
with respect to reflexive Kripke models.
\item[$(iii)$]
{\em \mvkfour{}} is sound and (strongly) complete
with respect to transitive Kripke models.
\item[$(iv)$]
{\em \mvsfour{}} is sound and (strongly) complete
with respect to reflexive and transitive (preordered) Kripke models.
\item[$(v)$]
{\em \mvb{}} is sound and (strongly) complete with
respect to symmetric Kripke models.
\item[$(vi)$]
{\em \mvsfive{}} is sound and (strongly) complete
with respect to reflexive and Euclidean Kripke models.\footnote{
This is the class of all Kripke models
whose accessibility relation is an equivalence relation.}
\end{itemize}
\end{theorem}


\section{Decidability of \mvk{} and its extensions}
\label{s: fmp}

In what follows, $\bL$ can be any of the
logics \mvk{},\mvd{},\mvt{},\mvkfour{},\mvsfour{},\mvb{} or \mvsfive{} and
$\bC_{\hspace{-0.1 em} \sbl}$ is the class of
the respective Kripke models, see Theorem \ref{t: extensions}.

We show that $\bL$ possesses the finite model property.
The proof is based on
the filtration technique, cf.~\cite[Chapter I, Section 7]{Segerberg71},
where this technique is applied to some two-valued modal logics.

Let $\Phi$ be a subformula-closed set of formulas\footnote{
That is,
if $\varphi \in \Phi$,
then each subformula of $\varphi$ also belongs to $\Phi$.}
and 
let ${M=\langle W,R,I \rangle}$
be a Kripke model.
The equivalence relation $\equiv_\Phi$ on $W$ is defined as follows.
\begin{center}
$u \equiv_\Phi v$
if and only if
$I(u,\varphi) = I(v,\varphi)$ for all $\varphi \in \Phi$. 
\end{center}

The \emph{$\bL$-filtration of $M$ through $\Phi$}
is the Kripke model
$M^\star_{\sbl,\Phi}
=
\langle W_{\sbl,\Phi}^\star,
R_{\sbl,\Phi}^\star,
I_{\sbl,\Phi}^\star \rangle$,
where

\begin{itemize}
\item
$W_{\sbl,\Phi}^\star$ is the set of
all equivalence classes of $\equiv_\Phi$.
That is, ${W_{\sbl,\Phi}^\star = \{ [u] : u \in W \}}$
where $[u]$ is the $\equiv_\Phi$ equivalence class of $u$.
\item
For $[u] \in W^\star$ and a propositional variable $p \in \Phi$,
$I_{\sbl,\Phi}^\star([u],p) = I(u,p)$.
By the definition of $\equiv_\Phi$, $I_{\sbl,\Phi}^\star$ is well defined
and
the value of $I_{\sbl,\Phi}^\star$
for $p \notin \Phi$ does not matter for our purposes. 
\item
The definition of $R_{\sbl,\Phi}^\star$ depends on $\bL$.
\begin{itemize}
\item
For \mvk{},\mvd{} and \mvt{}, $[u] R_{\sbl,\Phi}^\star [v]$
if and only if
there exist $u^\prime \in [u]$ and $v^\prime \in [v]$ such that
$u^\prime R v^\prime$.
\item
For \mvkfour{}, $[u]R_{\sbl,\Phi}^\star[v]$ if and only if
\begin{itemize}
\item
for all $\Box\varphi^\prime \in \Phi$,
$I(u,\Box\varphi^\prime) \leq I(v,\Box\varphi^\prime)$
and
$I(u,\Box\varphi^\prime) \leq I(v,\varphi^\prime)$; and
\item
for all $\Diamond\varphi^\prime \in \Phi$,
$I(u,\Diamond\varphi^\prime) \geq I(v,\Diamond\varphi^\prime)$
and
$I(u,\Diamond\varphi^\prime) \geq I(v,\varphi^\prime)$.
\end{itemize} 
\item
For \mvsfour{}, $[u]R_{\sbl,\Phi}^\star[v]$ if and only if
\begin{itemize}
\item
for all $\Box\varphi^\prime \in \Phi$,
$I(u,\Box\varphi^\prime) \leq I(v,\Box\varphi^\prime)$; and
\item
for all $\Diamond\varphi^\prime\in\Phi$,
$I(u,\Diamond\varphi^\prime) \geq I(v,\Diamond\varphi^\prime)$.
\end{itemize} 
\item 
For \mvb{}, $[u]R_{\sbl,\Phi}^\star[v]$ if and only if
\begin{itemize}
\item
for all $\Box\varphi^\prime \in \Phi$,
$I(u,\Box\varphi^\prime) \leq I(v,\varphi^\prime)$
and
$I(v,\Box\varphi^\prime) \leq I(u,\varphi^\prime)$; and
\item
for all $\Diamond\varphi^\prime \in \Phi$,
$I(u,\Diamond\varphi^\prime) \geq I(v,\varphi^\prime)$
and
$I(v,\Diamond\varphi^\prime) \geq I(u,\varphi^\prime)$.
\end{itemize} 
\item
For \mvsfive{}, $[u]R_{\sbl,\Phi}^\star[v]$ if and only if
\begin{itemize}
\item
for all $\Box\varphi^\prime \in \Phi$,
$I(u,\Box\varphi^\prime)= I(v,\Box\varphi^\prime)$ and
\item
for all $\Diamond\varphi^\prime \in \Phi$,
$I(u,\Diamond\varphi^\prime) = I(v,\Diamond\varphi^\prime)$.
\end{itemize} 
\end{itemize}
\end{itemize}

\begin{theorem}
\label{t: filtration} 
Let $M$ be in $\bC_{\hspace{-0.1 em} \sbl}$ and
let $M_{\sbl,\Phi}^\star$ be its $\bL$-filtration through $\Phi$.
Then
\begin{itemize}
\item
For all $\varphi \in \Phi$ and $u \in W$,
$I(u,\varphi) = I_{\sbl,\Phi}^\star([u],\varphi)$ and
\item
$M_{\sbl,\Phi}^\star$ is in $\bC_{\hspace{-0.1 em} \sbl}$.
\end{itemize}
\end{theorem}


\begin{definition}
\label{d: fmp}
A logic $\bL$ possesses the \emph{finite model property},
if for each finite set of sequents $\bsigma$ and
each sequent $\Gamma \rightarrow \Delta$ such that
$\bsigma \not\vdash_{\sbl} \Gamma \rightarrow \Delta$,
there exists a finite Kripke model
$M \in \bC_{\hspace{-0.1 em} \sbl}$
(i.e. the set of worlds of $M$ is finite) such that
$M \models \bsigma$, but $M \not\models \Gamma \rightarrow \Delta$.
\end{definition}

\begin{theorem}
\label{t: fmp}
Each of the logics considered above
possesses the finite model property.
\end{theorem}


\begin{corollary}
\label{c: decidability}
Each of the logics considered above is strongly decidable.
\end{corollary}

\begin{proof}
The decision procedure is standard.
We, in parallel, search
for a proof of $\Gamma \rightarrow \Delta$ from $\bsigma$
and for a finite Kripke model provided by Theorem~\ref{t: fmp}
that satisfies $\bsigma$,
but does not satisfy $\Gamma \rightarrow \Delta$.
\end{proof}

\section{Duality of $\Box$ and $\Diamond$ via negation}
\label{s: duality}

In \mvk{}, the existence of any specific connective is not assumed and
$\Diamond$ is not defined as
the De Morgan dual $\neg \Box \neg$ of $\Box$,
but is defined independently,
both semantically and syntactically via the proof system. 

In this section we define the truth table for negation $\neg$
in such a way that $\Box$ and $\Diamond$ become the De Morgan dual.
That is, the sequents
\begin{equation}
\label{eq: duality 1}
(\Diamond\varphi,k)\rightarrow (\neg\Box\neg\varphi,k)
\end{equation}
and
\begin{equation}
\label{eq: duality 2}
(\Box\varphi,k)\rightarrow (\neg\Diamond\neg\varphi,k)
\end{equation}
are provable in \mvk.\footnote{
In particular, in the three-valued logics of
{\L}ukasiewicz~\cite{Lukasiewicz20} and Kleene~\cite{Kleene38},
these connectives are interdefinable.}
We shall show that this is the only appropriate definition of negation,
for which~\eqref{eq: duality 1} and~\eqref{eq: duality 2}
are derivable in \mvk.

The truth table of $\neg$ is
\begin{equation}
\label{eq: neg}
\neg(\bv_k) = \bv_{n-k+1} \, \ \ \ \ k = 1,2,\ldots,n
\end{equation}
That is,
\[ 
\neg (\bv_1) = \bv_n, \ \neg (\bv_2) = \bv_{n-1}, \ \ldots, \ 
\neg (\bv_{n-1}) = \bv_2, \ \mbox{and} \ \neg (\bv_n)=\bv_1
\]
Therefore, axioms~\eqref{eq: table axiom} for $\neg$ are
\[
(\varphi,k) \rightarrow (\neg \varphi , n - k + 1)
\]

\begin{example}
\label{e: neg}
Sequents
\begin{equation}
\label{eq: neg example}
(\neg \varphi , n - k + 1) \rightarrow (\varphi,k)
\end{equation}
are \mvk{} derivable.

The derivation is as follows.
\[
\hspace{-3.1 em}
\begin{array}{lll}
1_{j \neq k}. & 
(\varphi,j) \rightarrow (\neg \varphi, n - j + 1) , \ \ j \neq k
&
\mbox{axiom~\eqref{eq: table axiom}}
\\
2_{j \neq k}. & 
(\varphi,j),(\neg \varphi,n - k + 1) \rightarrow
&
\mbox{follows from $1_j$ by
$n - j + 1,n - k + 1$-right-shift~\eqref{eq: mvml rs}}
\\
3. & 
(\neg \varphi , n - k + 1) \rightarrow (\varphi,k)
&
\mbox{follows from $2_{j \neq k}$
by multi-shift~\eqref{eq: mvml multi-shift}}
\end{array}
\]
\end{example}

\begin{remark}
\label{r: reversal}
Sequents~\eqref{eq: duality 1} and~\eqref{eq: duality 2}
immediately imply their reversals.
For~\eqref{eq: duality 1}, since each sequent in the set
\[
\{ (\Diamond\varphi,k^\prime) \rightarrow (\neg\Box\neg\varphi,k^\prime) :
k^\prime \neq k\}
\]
is derivable, by right shifts, we derive
\[
\{ (\Diamond\varphi,k^\prime),(\neg\Box\neg\varphi,k) \rightarrow :
k^\prime \neq k \}
\]
from which, by multi-shift, we obtain
\[
(\neg\Box\neg\varphi,k) \rightarrow (\Diamond\varphi,k)
\]
and,
dually, for~\eqref{eq: duality 2}. 
\end{remark}

\begin{theorem}
\label{t: negation}
Let $\neg$ be a unary connective.
Then, sequents~\eqref{eq: duality 1} and~\eqref{eq: duality 2}
are derivable in {\em \mvk} if and only if,
for all $k = 1,2,\ldots,n$, $\neg(\bv_k) = \bv_{n-k+1}$.
\end{theorem}


\begin{remark}
\label{r: redandant}
If we define negation as above,
then rule~\eqref{eq: mvk 2} becomes redundant, which can be shown
as follows.
\[
\begin{array}{lll}
1. & 
(\varphi,k) \rightarrow \Gamma^\times , \ \ k \neq 1
&
\mbox{assumption of~\eqref{eq: mvk 2}}
\\
2. & 
(\neg \varphi , n - k + 1) \rightarrow (\varphi,k)
&
\mbox{\eqref{eq: neg example}}
\\
3. & 
(\neg \varphi , n - k + 1) \rightarrow \Gamma^\times , \ \ n - k + 1 \neq n
&
\mbox{follows from 1 and 2 by cut}
\\
4. & 
(\Box\neg\varphi,n -k +1),\Gamma \rightarrow , \ \ n - k + 1\neq n
&
\mbox{follows from 3 by~\eqref{eq: mvk 1}}
\\
5. & 
(\neg\Box\neg\varphi, k) \rightarrow (\Box\neg\varphi,n - k + 1)
&
\mbox{\eqref{eq: neg example}}
\\
6. & 
(\neg \Box \neg \varphi, k),\Gamma \rightarrow , \ \ k \neq 1
&
\mbox{follows from 4 and 5 by cut}
\\
7. & 
\eqref{eq: mvk 2}
&
\mbox{because, by~\eqref{eq: duality 1}, $\neg \Box \neg$ is $\Diamond$}
\end{array}
\]

Also, it can be shown
that~\eqref{eq: mvml reflexive 2}, \eqref{eq: mvk transitive 2},
\eqref{eq: mvk symmetric 2}, and~\eqref{eq: mvk euclidean 2} follow
from~\eqref{eq: mvml reflexive 1}, \eqref{eq: mvk transitive 1},
\eqref{eq: mvk symmetric 1}, and~\eqref{eq: mvk euclidean 1},
respectively, and vice-versa.
\end{remark}

\section{Embedding many-valued intuitionistic logic into \mvsfour}
\label{s: intuitionistic}

In \cite{Takano86}, following \cite{Rousseau70},
Takano defined a quite general notion of many-valued intuitionistic logic,
that we shall denote by \mvil{}.
We focus on the semantics, because we embed \mvil{} into \mvsfour{}
semantically.
Also, we restrict ourselves to the case of linearly ordered set of
truth values $\bV$ in which
\mvil{}-interpretations may be defined, recursively, as follows.

The language of \mvil{} is that of many-valued propositional logic,
i.e., it does not contain the modal connectives~$\Box$ or $\Diamond$.

An \mvil{}-interpretation ${M} = \langle W,R,{I} \rangle$
is a preordered (reflexive and transitive) many-valued Kripke model
satisfying the (monotonic valuation) requirement below.

For all propositional variables $p \in \P$ and
for all $u,v \in W$ such that $uRv$,
\[
{I}(u,p) \leq {I}(v,p)
\]

The definition of $I$ extends to formulas of
the form $\ast(\varphi_1,\ldots,\varphi_\ell)$ as
\begin{equation}
\label{eq: extension}
{I}(u,\ast(\varphi_1,\ldots,\varphi_\ell))
=
\inf
\{\ast({I}(v,\varphi_1),\ldots,{I}(v,\varphi_\ell)) : v \in S(u)\}
\end{equation}

A straightforward induction on the formula complexity shows
that ${I}$ is monotonic not only on $W\times \P$,
but on the whole $W \times \F$.

We write ${M},u \models_\mvil{} (\varphi,k)$,
if ${I}(u,\varphi) = \bv_k$.
For a sequent $\Gamma \rightarrow \Delta$ and a set of sequents $\bsigma$,
we define the relations
${M},u \models_\mvil{} \Gamma \rightarrow \Delta$,
${M} \models_\mvil{} \Gamma \rightarrow \Delta$,
${M} \models_\mvil{} \bsigma$, and
$\bsigma\models_\mvil{} \Gamma \rightarrow \Delta$
like in the beginning of Section~\ref{s: logic}.

Our translation of \mvil{} to \mvsfour{}, is a generalization of
the two-valued case (first suggested in~\cite{Godel33}).

\begin{definition}
\label{d: translation}
Let $\varphi$ be a formula in the language of \mvil{}.
The translation $\varphi^\t$ of an \mvil{} formula $\varphi$ is obtained
from $\varphi$ by inserting $\Box$ before every its subformula.
That is, $\varphi^\t$ is defined recursively as follows.
\begin{itemize}
\item
For a propositional variable $p$, $p^\t$ is $\Box p$, and
\item
if $\varphi$ is of the form $\ast(\varphi_1,\ldots,\varphi_\ell)$,
then $\varphi^\t$ is $\Box\ast(\varphi_1^\t,\ldots,\varphi_\ell^\t)$.
\end{itemize}
\end{definition}

\begin{lemma}
\label{l: intuitionistic lemma}
Let $M = \langle W,R,I \rangle$ be a preordered Kripke model and
let $\widehat{M} = \langle W,R,\widehat{I} \rangle$ be such that,
for all $u \in W$ and all $p \in \P$,
$\widehat{I}(u,p) = I(u,\Box p)$.
Then $\widehat{M}$ is an {\em \mvil{}}-interpretation, and,
for all $u \in W$ and all formulas $\varphi$ in
the language of {\em \mvil{}},
\begin{equation}
\label{eq: for il}
\widehat{I}(u,\varphi) = I(u,\varphi^t)
\end{equation}
\end{lemma}

\begin{proof}
To show $\widehat{M}$ is an \mvil{}-interpretation,
we need to show that,
for all $u,v \in W$ such that $uRv$ and for all $p\in \P$,
$\widehat{I}(u,p) \leq \widehat{I}(v,p)$,
i.e., by the definition of $\widehat{I}$,
we need to show $I(u,\Box p) \leq I(v,\Box p)$,
which is clear, because $M$ is transitive.

The proof of~\eqref{eq: for il} is by induction on
the complexity of $\varphi$
(extending $\widehat{I}$ to an intuitionistic valuation).

The basis, i.e., the case of $\varphi$ being a propositional variable,
is by the definition of $\widehat{I}$, and,
for the induction step, if $\varphi$ is of the form
$\ast(\varphi_1,\ldots,\varphi_\ell)$,
then
\begin{eqnarray*}
\widehat{I}(u,\varphi)
& = &
\inf\{\ast(\widehat{I}(v,\varphi_1),\ldots,\widehat{I}(v,\varphi_\ell)) :
v \in S(u) \}
\\
& = &
\inf\{\ast(I(v,\varphi_1^\t),\ldots,I(v,\varphi_\ell^\t)) : v \in S(u) \}
\\
& = &
\inf\{I(v,\ast(\varphi_1^\t,\ldots,\varphi_\ell^\t) : v \in S(u) \}
\\
& = &
I(u,\Box\ast(\varphi_1^\t,\ldots,\varphi_\ell^\t))
\\
& = &
I(u,\varphi^\t)
\end{eqnarray*}
where the first equality is by~\eqref{eq: extension},
the second equality is by the induction hypothesis,
the third and the fourth equalities are by
the definition of the extension of $I$ onto~$W \times \F$, and
the last equality is by the definition of translation $^t$.
\end{proof}

It follows from~\eqref{eq: for il} that
$\widehat{M} \models_\mvil{} \Gamma \rightarrow \Delta$
if and only if
$M \models \Gamma^\t \rightarrow \Delta^t$,
where $\Gamma^\t$ and $\Delta^t$ are obtained from
$\Gamma$ and $\Delta$, respectively,
by translating every formula appearing in them.
Similarly, $\widehat{M} \models_\mvil{} \bsigma$ if and only if
$M \models \bsigma^\t$ where $\bsigma^\t$ is obtained from $\bsigma$
by translating every sequent appearing in it.

\begin{theorem}
{\em $\bsigma \models_\mvil \Gamma \rightarrow \Delta$}
if and only if
$\bsigma^\t \models_\mathbf{C} \Gamma^t \rightarrow \Delta^t$,
where $\mathbf{C}$ is the class of preordered Kripke models.
\end{theorem}

\begin{proof}
If $\bsigma^\t \not\models_\mathbf{C} \Gamma^\t \rightarrow \Delta^t$,
there exists a preordered Kripke model $M$ such that
$M \models \bsigma^t$,
but $M \not\models \Gamma^\t \rightarrow \Delta^t$.
By Lemma~\ref{l: intuitionistic lemma},
$\widehat{M} \models_\mvil \bsigma$,
but $\widehat{M} \not\models_\mvil \Gamma \rightarrow \Delta$.
Thus, $\bsigma\not\models_\mvil \Gamma \rightarrow \Delta$.

Conversely,
if $\bsigma \not\models_\mvil \Gamma \rightarrow \Delta$,
there exists an \mvil{} interpretation $M$ such that
$M \models \bsigma$,
but $M \not\models \Gamma \rightarrow \Delta$.
By definition, $M$ is also a preordered Kripke model and
$\widehat{M}$ defined in Lemma~\ref{l: intuitionistic lemma} is
$M$ itself, because by the definition of an intuitionistic valuation,
the value of a propositional variable $p$ in
a world $u$ is already the minimum of the values of $p$ in $S(u)$.
Therefore $M$, as an \mvsfour{} model,
satisfies $\bsigma^\t$ but not $\Gamma^\t \rightarrow \Delta^t$.
\end{proof}

It follows that strong decidability (and completeness) of \mvsfour{}
implies strong decidability of \mvil{}.

\begin{remark}
\label{r: monotonic connective}
If the principal connective $\ast$ of a formula is monotonic,\footnote{
That is, if
$\bv_{k_1} \leq \bv_{k_1^\prime},\ldots,\bv_{k_\ell}
\leq
\bv_{k_\ell^\prime}$,
then
$\ast(\bv_{k_1},\ldots,\bv_{k_\ell})
\leq
\ast(\bv_{k_1^\prime},\ldots,\bv_{k_\ell^\prime})$. For example, in the three-valued logics of
{\L}ukasiewicz~\cite{Lukasiewicz20} and Kleene~\cite{Kleene38}, disjunction $\vee$ and conjunction $\wedge$ are monotonic.}
then there is no need to insert $\Box$ before $\ast$ in the translation.
This is because $\widehat{I}$ is ``local'' on this connective,
like in modal logic.
\end{remark}

\bibliographystyle{eptcs}
\bibliography{11}

\end{document}